\documentclass[11pt]{article}

\usepackage[T1]{fontenc}
\usepackage{lmodern}
\usepackage[margin=1in]{geometry}
\usepackage{authblk}
\usepackage{xcolor}

\usepackage{amsmath,amssymb,amsthm,mathtools}
\usepackage{microtype}
\usepackage{natbib}
\usepackage{comment}
\usepackage{tikz}
\usepackage{graphicx}
\usepackage{pgfplots}
\usepgfplotslibrary{groupplots}
\pgfplotsset{compat=1.18}
\usetikzlibrary{arrows.meta,calc}
\usepackage[hidelinks]{hyperref}

\allowdisplaybreaks

\newtheorem{theorem}{Theorem}
\newtheorem{lemma}{Lemma}

\newtheorem{definition}{Definition}

\newtheorem*{theorem*}{Theorem}

\DeclareMathOperator{\Unif}{Unif}
\newcommand{\OPT}{\operatorname{OPT}}
\newcommand{\OPTpart}{\operatorname{OPT}_{\mathrm{part}}}
\newcommand{\E}{\mathbb E}
\newcommand{\R}{\mathbb R}
\newcommand{\conv}{\operatorname{conv}}
\newcommand{\downconv}{\operatorname{downconv}}
\newcommand{\supp}{\operatorname{supp}}
\newcommand{\diff}{\,\mathrm{d}}

\title{No Extra Signals Needed:\\The Uniform Price of Explainable Information Design}
\author[1]{Francesco Bacchiocchi}
\author[2]{Tommaso Cesari}
\author[3]{Roberto Colomboni}

\affil[1]{DEIB, Politecnico di Milano, Milano, Italy}
\affil[2]{School of Electrical Engineering and Computer Science, University of Ottawa, Ottawa, Canada}
\affil[3]{School of Mathematics, University of Bristol, Bristol, United Kingdom}

\affil[ ]{\footnotesize
\texttt{francesco.bacchiocchi@polimi.it},
\texttt{tcesari@uottawa.ca},
\texttt{roberto.colomboni@bristol.ac.uk}
}
\date{}

\begin{document}
\maketitle

\begin{abstract}
In information design, an informed sender aims to influence a receiver’s decision by committing to a signaling scheme. However, optimal signaling schemes often rely on randomization or assign the same signal to disconnected regions of the state space, making them difficult to interpret or communicate. Motivated by these limitations, we focus on explainable information design in the one-dimensional linear setting, where an explainable policy partitions the state space into at most $K$ consecutive intervals and deterministically sends a distinct signal for each interval. We study the \emph{price of explainability}, defined as the worst-case ratio between the optimal value achieved by an explainable signaling scheme and that achieved by an unrestricted signaling scheme using the same number of signals. Under a uniform prior, \cite{ChenLinTangTuckerFoltz2026} established a tight $2/3$ guarantee when the explainable signaling scheme was allowed to use additional signals. They also showed that the same $2/3$ guarantee holds when both the explainable and unrestricted signaling schemes use at most $K$ signals, provided that utilities are binary-valued and $K \geq 4$, leaving the case of arbitrary bounded utilities open.

We resolve this question completely. Under a uniform prior, the price of explainability is exactly $1/2$ for $K=2$ and exactly $2/3$ for every $K \geq 3$. For both regimes, we also show that the corresponding ratios are tight.
\end{abstract}

\section{Introduction}

Information design, also known as Bayesian persuasion, was introduced by~\citet{KamenicaGentzkow2011} to study how an informed sender can strategically disclose private information in order to influence a receiver’s behavior. In this framework, the sender’s private information is represented by the state of nature. Before the state is realized, the sender commits to an information policy (also called a signaling scheme), \emph{i.e.}, a mapping from states of nature to probability distributions over signals. Once the state is realized and observed by the sender, a signal is generated according to the committed signaling scheme. The receiver then updates their prior beliefs about the underlying state and chooses an optimal action. Information design has become a fundamental tool in economics and algorithmic game theory, with applications to auctions, recommendation systems, voting, contracting, and market design.

However, the structure of an optimal information policy is often opaque. Indeed, even when the state space is a one-dimensional interval and the policy is deterministic, it may still lack explainability, as the same signal may be assigned to disconnected subsets of the state space (see, \emph{e.g.}, Example~1.1 in \citet{ChenLinTangTuckerFoltz2026}). Although such non-monotone policies can be optimal \citep{ArieliBabichenkoSmorodinskyYamashita2023}, they may be undesirable in practice because of transparency concerns and the difficulty of understanding them. This motivates the study of information policies with simpler and more interpretable structures.

\citet{ChenLinTangTuckerFoltz2026} proposed the framework of
\emph{explainable information design} to formalize this requirement. In the one-dimensional setting, an explainable policy chooses cutoffs that partition the state space into at most $K$ consecutive intervals and deterministically sends a distinct signal for each interval, thereby revealing which interval contains the realized state. They define the \emph{price of explainability} as the {worst-case} ratio of the optimal expected utility of an explainable policy to that of an unrestricted policy, both using at most $K$ signals.

For absolutely continuous one-dimensional priors with full support and the same signal budget on both sides, \citet{ChenLinTangTuckerFoltz2026} proved that this ratio is exactly $1/2$. Indeed, the best explainable policy always retains at least half of the unrestricted optimum, and no larger universal guarantee is possible. Under a uniform prior, they obtained the stronger ratio $2/3$ in two settings: for binary-valued utilities when both policies use at most $K$ signals and $K \ge 4$, and for general utilities when the explainable policy may use $K'\ge3K/2$ signals against an unrestricted policy using at most $K$ signals.
They left open whether, under a uniform prior and for general utilities, the $2/3$ guarantee can be achieved when both policies use at most $K$ signals.

We completely resolve this open question for arbitrary bounded utilities. For $K=2$, the price of explainability is $1/2$, whereas for every $K\ge 3$, it is $2/3$.

\subsection{Setting}\label{sec:setting}

In our framework, the state of nature, denoted by
$\theta\in\Theta$, is drawn from the uniform prior distribution on
$\Theta\coloneqq[0,1]$. A signaling scheme over a finite signal set
$\mathcal S$, with $|\mathcal S|\leq K$, is a Markov kernel from
$\Theta$ to $\mathcal S$. That is, it is a map
\[
    \sigma:\Theta\times\mathcal S\to[0,1]
\]
such that, for every $s\in\mathcal S$, the function
$\theta\mapsto\sigma(s\mid\theta)$ is measurable and, for every
$\theta\in\Theta$,
\[
    \sum_{s\in\mathcal S}\sigma(s\mid\theta)=1.
\]
Here, $\sigma(s\mid\theta)$ denotes the probability of sending signal
$s$ after observing state $\theta$. We say that a signaling scheme is
\emph{deterministic} if, for every $\theta\in\Theta$, there exists a
signal $s\in\mathcal S$ such that $\sigma(s\mid\theta)=1$.

The ex ante probability of observing signal $s$ is
\[
    \Pr[s]
    =
    \int_0^1 \sigma(s\mid\theta)\diff\theta.
\]
We discard any signal $s$ such that $\Pr[s]=0$. Hence, without
loss of generality, every signal in $\mathcal S$ has positive ex ante
probability. For every $s\in\mathcal S$, the receiver's posterior mean is
\[
    \mu_s
    \coloneqq
    \E[\theta\mid s]
    =
    \frac{\int_0^1 \theta\,\sigma(s\mid\theta)\diff\theta}
         {\Pr[s]}.
\]
In this paper, we focus on \emph{linear information design}, where the
sender's payoff depends on the posterior belief only through its mean.
Thus, when signal $s$ induces posterior mean $\mu_s$, the sender
receives payoff $u(\mu_s)$, where $u:[0,1]\to[0,1]$.
Thus, the expected value of the signaling scheme is
\[
    \sum_{s\in\mathcal S}\Pr[s]\,u(\mu_s).
\]
Every signaling scheme satisfies the Bayes-plausibility conditions
\[
    \sum_{s\in\mathcal S}\Pr[s]=1
    \qquad\text{and}\qquad
    \sum_{s\in\mathcal S}\Pr[s]\mu_s
=
    \E[\theta]
=
    \frac{1}{2}.
\]
We write $\OPT^u(K)$ for the supremum value over all signaling
schemes with at most $K$ signals. When the utility is clear, we
simply write $\OPT(K)$.

An \emph{interval-partitional} signaling scheme is a deterministic
signaling scheme specified by thresholds
$0=a_0\leq a_1\leq\cdots\leq a_K=1$, which assigns a distinct signal
to the interior of each nonempty interval
$[a_{i-1},a_i]$.\footnote{At the threshold points
$\{a_i\}_{i=0}^K$, the signaling scheme may send an arbitrary signal. Under a
uniform prior, these points have probability zero and therefore do
not affect the information policy.}
Under the uniform prior, each nonempty interval
$[a_{i-1},a_i]$ has mass $a_i-a_{i-1}$ and posterior mean $({a_{i-1}+a_i})/{2}.$
We do not exclude the possibility of zero-length intervals, so optimizing over schemes with at most $K$ nonempty intervals is equivalent to optimizing over $K$ intervals, some of which may be empty. We denote by $\OPTpart^u(K)$ the supremum value over all interval-partitional
signaling schemes with at most $K$ nonempty intervals.
When the utility is clear, we simply write $\OPTpart(K)$.

For $K\geq 2$, we define the \emph{uniform price of explainability}, following \citet{ChenLinTangTuckerFoltz2026}, as follows
\[
    c_K^{\rm unif}
    :=
    \inf_{u:[0,1]\to[0,1]}
    \frac{\OPTpart^u(K)}{\OPT^u(K)},
\]
ignoring utilities with $\OPT^u(K)=0$. Equivalently, our lower bounds are payoff inequalities that remain meaningful even when the denominator is zero.

\subsection{Original contributions}\label{sec:contributions}

\paragraph{Overview of the results.} Our main contribution is an exact characterization of the uniform price of explainability. Precisely: 
\begin{itemize}
    \item For every $K\ge3$ and every bounded utility
    $u:[0,1]\to[0,1]$, we prove
    \[
        \OPTpart(K)\ge\frac23\OPT(K).
    \]
    For $K=3$, this result is new even for binary-valued utilities, improving on the guarantee established by \citet{ChenLinTangTuckerFoltz2026}, while for $K\ge4$, it extends the guarantee of
    \citet{ChenLinTangTuckerFoltz2026} from binary-valued utilities to
    arbitrary bounded utilities.

    \item We show that $K=2$ is the unique exception by constructing
    binary-valued uniform-prior instances whose ratio approaches
    $1/2$.
    Combined with the general $1/2$ lower bound of
    \citet{ChenLinTangTuckerFoltz2026}, this yields
    $c^{\rm unif}_2=1/2$.
\end{itemize}

The matching upper bound $c_K^{\rm unif}\le2/3$ for $K\ge3$ already
follows from the example of
\citet{ChenLinTangTuckerFoltz2026} in which the sender's utility is supported at
posterior means $1/3$ and $2/3$. We summarize our results in the following theorem.

\begin{theorem*}[{Informal version of Theorem~\ref{thm:main}}]
The uniform price of explainability satisfies
\[
    c_K^{\rm unif}
    =
    \begin{cases}
        \dfrac12, & K=2,\\[2mm]
        \dfrac23, & K\ge3.
    \end{cases}
\]
\end{theorem*}

\paragraph{Techniques and challenges.}\label{sec:techniques}

\cite{ChenLinTangTuckerFoltz2026} showed that, under a uniform prior, an explainable policy can preserve at least two thirds of the value of an unrestricted $K$-signal policy for arbitrary utilities, provided that the explainable policy may use $K' \geq 3K/2$ intervals. Their proof builds on a result by~\cite{ArieliBabichenkoSmorodinskyYamashita2023}, which allows an optimal finite-signal policy to be decomposed into consecutive intervals, each associated with at most two posterior means. \cite{ChenLinTangTuckerFoltz2026} convert each interval in this decomposition separately into at most three consecutive subintervals. Taken together, these subintervals form the final explainable signaling scheme. This interval-by-interval conversion explains why their construction may require the larger budget $K' \geq 3K/2$.

With only $K$ intervals, these conversions cannot be performed independently. A choice made in one part of the policy may reduce the number of intervals available for the remaining parts. Theorem~\ref{thm:k2} already illustrates this difficulty. Even when the unrestricted policy induces only two posterior means, there may be no partition into two intervals that preserves both of them. 

\cite{ChenLinTangTuckerFoltz2026} also proved that the uniform price of explainability is $2/3$ for binary-valued utilities when $K\geq 4$. Extending this guarantee to arbitrary utilities, however, presents an additional challenge. Under a binary utility, each posterior mean is simply classified as useful or not useful, depending on whether it gives the sender positive utility. It is therefore enough to preserve sufficient probability on the set of useful posterior means. For arbitrary utilities, different posterior means can contribute different amounts to the payoff, so preserving enough total probability does not necessarily preserve enough value. We therefore need a stronger construction that simultaneously preserves a fraction of the probability associated with each posterior mean.

We start from an arbitrary finite-signal information policy and merge signals that induce the same posterior mean. The policy can then be represented by an ordered target set $S=\{x_1,\ldots,x_m\}$ and a probability vector $(p_1,\ldots,p_m)$, where $p_i$ is the probability assigned to posterior mean $x_i$. Notice that $m\leq K$, since the number of distinct posterior
means cannot exceed the number of signals used by the policy. Our goal is to show that there exists a distribution over explainable policies under which the expected probability assigned to each posterior mean \(x_i\) is at least \(2p_i/3\). 

We do not construct such a distribution directly. Instead, we introduce a random centered packing, namely, a random, possibly incomplete collection of disjoint intervals centered at the posterior means that we want the final policy to induce. Under the uniform prior, the length of an interval centered at $x_i$ is the probability assigned to that posterior mean by an interval-partitional policy.
We therefore aim to guarantee, for every $i$, that the interval centered at $x_i$ has expected length at least $2p_i/3$.

For a fixed ordered target set $S$, we consider the probability vectors induced by signaling schemes whose posterior means belong to $S$. These
vectors satisfy the Bayes-plausibility conditions~\citep{KamenicaGentzkow2011}
and additional constraints imposed by the uniform prior. The latter take the form of prefix inequalities, since each of them concerns an initial segment, or prefix, of the ordered target set. We denote by $F_m(S)$ the set of all probability vectors satisfying these conditions.

The set $F_m(S)$ is compact and convex. Our goal is to construct, for every $p\in F_m(S)$, a random centered packing indexed by the elements of $S$ whose expected interval length at each target $x_i$ is at least $2p_i/3$ and omits at least one element of $S$ in every realization. Rather than constructing such a packing directly, we reduce the problem to smaller target sets. Since every vector in $F_m(S)$ is a convex combination of its extreme points and the desired packing guarantees are preserved under mixtures, it is enough to consider an extreme point of $F_m(S)$. At such a point, one of two things happens. Either some probability $p_i$ is zero, in which case we remove the corresponding posterior mean, solve the smaller instance, and assign zero length to the removed coordinate, or at least one proper prefix inequality is tight. In the latter case, the target set and its probability vector split into two adjacent smaller instances. After an affine rescaling, we recursively construct a random centered packing for each smaller instance, using instances in which the target set has cardinality one or two as base cases. We then couple and concatenate the two resulting packings to obtain a centered packing on the original target set $S$ that omits at least one element of $S$ in every realization. Finally, mixing the constructions associated with the extreme points gives the desired packing for every $p\in F_m(S)$.

Hence, if the target set $S$ has cardinality $m$, every realization of the random packing uses intervals centered at no more than $m-1$ elements of $S$, and therefore at no more than $K-1$, since $m\leq K$. This support bound allows us to apply a convex-hull argument showing that the expected length vector of the random packing is coordinatewise dominated by a convex combination of length vectors induced by explainable signaling schemes using at most $K$ intervals. Since the utility is nonnegative, at least one deterministic partition in this combination achieves at least two thirds of the payoff of the original unrestricted policy, proving the desired result.

\subsection{Related work}\label{sec:related-work}

Bayesian persuasion, introduced by \citet{KamenicaGentzkow2011}, studies how an informed sender can strategically disclose information to influence the decision of a Bayesian receiver. Over the past decade, this framework has received considerable attention from the computer science community and has been applied to a wide range of settings, including ad auctions \citep{EmekFeldmanGamzuPaesLemeTennenholtz2014,BadanidiyuruBhawalkarXu2018,Bacchiocchi2022Public}, price discrimination and fair selection \citep{BanerjeeMunagalaShenWang2024, BanerjeeMunagalaShenWang2025}, principal-agent and contract-design settings \citep{BabichenkoTalgamCohenXuZabarnyi2024,CastiglioniChen2025}, online recommendation and misinformation \citep{ZuIyerXu2021,FengTangXu2022,HossainMladenovicChenGidel2024}, and voting \citep{AlonsoCamara2016a,AlonsoCamara2016b}.

\paragraph{Linear information design.} 

Our work belongs to the rich literature on one-dimensional linear information design. In this setting, uncertainty is represented by a single numerical state, and the sender’s payoff depends on the receiver’s posterior belief only through its mean. \citet{GentzkowKamenica2016} and \citet{KleinerMoldovanuStrack2021} characterize which distributions of posterior means can be induced by an information policy while remaining consistent with the prior.
\citet{DworczakMartini2019}, \citet{ArieliBabichenkoSmorodinskyYamashita2023}, and \citet{KolotilinLiZapechelnyuk2025} investigate the structure of optimal policies and study conditions under which such policies can be chosen to have a simple, partitional, or monotone structure. These results are particularly relevant to our setting because an optimal policy need not be easy to interpret. The characterizations of \citet{ArieliBabichenkoSmorodinskyYamashita2023} and \citet{KleinerMoldovanuStrack2021}, for example, imply that an optimal policy can be organized into intervals within which it either fully reveals the state or induces at most two posterior means. Finally, \citet{FengHoTang2024} study settings with receivers whose behavior follows a quantal-response model.

\paragraph{Relation to \cite{ChenLinTangTuckerFoltz2026}.} Explainable information design was introduced by
\citet{ChenLinTangTuckerFoltz2026} to study information policies with a
simple and interpretable structure.
In the one-dimensional model, their explainable class consists of
deterministic policies that reveal which interval of a partition contains the realized state of nature.
They prove a tight price of explainability of $1/2$ for
absolutely continuous one-dimensional priors with full support, for every
$K\ge2$.
Under the uniform prior, they obtain a tight $2/3$-approximation when the
interval scheme may use $K'\ge3K/2$ signals, and leave open whether the
signal increase is necessary.
They also prove a $2/3$ ratio for binary-valued utilities
when $K\ge4$.
Our paper resolves this open question completely. For arbitrary bounded utilities, the price of explainability is exactly $1/2$ for $K=2$ and exactly $2/3$ for every $K \geq 3$.

\paragraph{Persuasion under communication and signal constraints.}
Our work is also related to Bayesian persuasion under communication and signal constraints. Indeed, while explainability restricts the structure of a
signaling scheme, a bound on the number of signals limits the amount of
information that the sender can communicate. \citet{GradwohlHahnHoeferSmorodinsky2022} analyze algorithms for persuasion with limited communication. \citet{LyuSuenZhang2024} consider coarse information design with a continuous state space and a finite signal space. Finally, \citet{ZhaoLiuShen2025} study how signaling schemes with a limited number of signals approximate the optimal unrestricted scheme.

\section{The two-signal case}\label{sec:k2}

The two-signal case is the only one in which the price of explainability under the uniform prior is $1/2$, rather than $2/3$.
A direct proof of the lower bound relies on a simple geometric observation.
Consider an arbitrary unrestricted signaling scheme and fix one of its signals. The probability with which this signal is sent cannot exceed the length of the longest interval $I \subseteq [0,1]$ whose midpoint is the posterior mean induced by that signal. This interval necessarily touches one of the endpoints of $[0,1]$, and therefore $I$ and its complement form an interval partition.
We can thus construct an interval-partitional signaling scheme using $I$ and its complement. Since $I$ induces the same posterior mean as the original signal and has length at least equal to the probability of that signal, it preserves at least its entire contribution to the sender’s expected utility.
Repeating the argument for the other signal yields a second interval-partitional signaling scheme. Choosing the better of the two signaling schemes preserves at least half the total payoff. 

To formalize the argument above, we first present the following lemma.

\begin{lemma}\label{lem:max-mass}
Let $q:[0,1]\to[0,1]$ be measurable, and set
\[
    \alpha:=\int_0^1q(t)\diff t.
\]
If $\alpha>0$ and
\[
    m:=\frac{1}{\alpha}\int_0^1tq(t)\diff t,
\]
then
\[
    \alpha\le2m\quad\text{if }m\le\frac12,
    \qquad
    \alpha\le2(1-m)\quad\text{if }m\ge\frac12.
\]
\end{lemma}

\begin{proof}
Among all functions $q:[0,1]\to[0,1]$ with integral $\alpha$, the first
moment is minimized by $\mathbf 1_{[0,\alpha]}$.  Indeed,
\begin{align*}
    \int_0^1tq(t)\diff t-\int_0^\alpha t\diff t
    =
    \int_\alpha^1tq(t)\diff t
    -\int_0^\alpha t(1-q(t))\diff t
    \ge
    \alpha\int_\alpha^1q(t)\diff t
    -\alpha\int_0^\alpha(1-q(t))\diff t
    =0,
\end{align*}
because the two integrals in the last line are equal.  Hence
$
    \alpha m\ge{\alpha^2}/{2},
$
which gives $\alpha\le2m$.  Applying the same argument to
$\widetilde q(t):=q(1-t)$ gives $\alpha\le2(1-m)$.
\end{proof}

Using the lemma above, we can prove the following theorem holds.

\begin{theorem}[{Price of explainability for $K=2$}]\label{thm:k2}
For every bounded utility $u:[0,1]\to[0,1]$,
\[
    \OPTpart^u(2)\ge\frac12\OPT^u(2).
\]
Moreover, for every $\varepsilon\in(0,1/12)$, there is a binary-valued
utility $u_\varepsilon:[0,1]\to \{0,1\}$ such that
\[
    \OPT^{u_\varepsilon}(2)=1,
    \qquad
    \OPTpart^{u_\varepsilon}(2)=\frac12+2\varepsilon.
\]
Consequently, $c_2^{\rm unif}=1/2$.
\end{theorem}
\begin{proof}
Consider an arbitrary signaling scheme using at most two signals with positive probability. If the signaling scheme uses only one signal $s$, then observing $s$ reveals no information about the realized state, so the posterior distribution coincides with the prior. Its posterior mean is
therefore $\mathbb{E}[\theta \mid s]={1}/{2}.$
Hence, the interval-partitional scheme consisting of the single interval $[0,1]$ induces the same posterior mean and attains the same payoff.

Suppose now that the signaling scheme uses two signals $s_1$ and $s_2$, both with positive probability. For each $i\in\{1,2\}$, let
$
p_i:=\Pr[s_i] $ and $
\mu_i:=\mathbb{E}[\theta\mid s_i].
$
Without loss of generality, we assume that $\mu_1<\mu_2$. Since the prior is uniform on $[0,1]$, for each $i\in\{1,2\}$ we have
\[
p_i=\int_0^1 \sigma(s_i\mid\theta)\,d\theta
\quad 
\textnormal{and}
\quad 
p_i\mu_i
=
\int_0^1 \theta\,\sigma(s_i\mid\theta)\,d\theta.
\]
Moreover,
\[
p_1\mu_1+p_2\mu_2
=
\int_0^1
\theta\bigl(\sigma(s_1\mid\theta)+\sigma(s_2\mid\theta)\bigr)
\,d\theta
=
\int_0^1\theta\,d\theta
=
\frac{1}{2}.
\]
Since $p_1,p_2>0$ and $\mu_1<\mu_2$, it follows that
$
\mu_1<{1}/{2}<\mu_2.
$

Let $u_i:=u(\mu_i)$ for each $i\in\{1,2\}$. Applying Lemma~\ref{lem:max-mass} to the
function $\theta\mapsto\sigma(s_1\mid\theta)$ and using
$\mu_1<1/2$ gives $p_1\leq 2\mu_1.$
Consider the interval-partitional signaling scheme induced by the
partition $[0,2\mu_1]$ and $[2\mu_1,1]$. The interval $[0,2\mu_1]$ has length $2\mu_1$ and posterior mean
$\mu_1$. It therefore contributes $2\mu_1u_1\geq p_1u_1$
to the sender's expected utility. Since $u$ is nonnegative, the other
interval contributes a nonnegative amount. Hence,
\begin{equation}
\mathrm{OPT}_{\mathrm{part}}^u(2)\geq p_1u_1.
\label{eq:two-signal-left}
\end{equation}
By a similar argument, applying Lemma~\ref{lem:max-mass} to the function
$\theta\mapsto\sigma(s_2\mid\theta)$ and using $\mu_2>1/2$ gives
$
p_2\leq 2(1-\mu_2).
$
Consider the interval-partitional signaling scheme induced by the
partition $[0,2\mu_2-1]$ and $[2\mu_2-1,1].$ The interval $[2\mu_2-1,1]$ has length $2(1-\mu_2)$ and posterior mean
$\mu_2$. Therefore,
\begin{equation}
\mathrm{OPT}_{\mathrm{part}}^u(2)\geq p_2u_2.
\label{eq:two-signal-right}
\end{equation}
Combining Equations~\eqref{eq:two-signal-left}
and~\eqref{eq:two-signal-right}, we obtain
\[
\mathrm{OPT}_{\mathrm{part}}^u(2)
\geq
\max\{p_1u_1,p_2u_2\}
\geq
\frac{1}{2}\bigl(p_1u_1+p_2u_2\bigr).
\]
Since the initial signaling scheme was chosen arbitrarily, taking the
supremum over all signaling schemes with at most two signals yields
\[
\mathrm{OPT}_{\mathrm{part}}^u(2)
\geq
\frac{1}{2}\,\mathrm{OPT}^u(2).
\]
To conclude the proof, fix $\varepsilon\in(0,1/12)$, set
\[
x_\varepsilon:=\frac14+\varepsilon,
\qquad
y_\varepsilon:=\frac34-\varepsilon,
\]
and define $u_\varepsilon(t):=\mathbf{1}\{t=x_\varepsilon\}+\mathbf{1}\{t=y_\varepsilon\}.$ Consider the signaling scheme that sends one signal when
$\theta\in[\varepsilon,1/2+\varepsilon]$ and another signal otherwise.
Each signal is observed by the receiver with probability $1/2$, and the
corresponding posterior means are $x_\varepsilon$ and $y_\varepsilon$.
Thus, $\OPT^{u_\varepsilon}(2)=1.$

An interval-partitional signaling scheme with two signals is determined
by a threshold $a\in[0,1]$, and its posterior means are $a/2$ and
$(1+a)/2$. The left interval induces a posterior mean providing positive
utility to the sender only when $a=2x_\varepsilon= 1/2 +2\varepsilon,$
while the right interval induces a posterior mean providing positive
utility to the sender only when $a=2y_\varepsilon-1= 1/2  -2\varepsilon.$

The two alternative matches are infeasible: inducing
$y_\varepsilon$ on the left interval would require
$a=2y_\varepsilon>1$, whereas inducing $x_\varepsilon$ on the
right interval would require $a=2x_\varepsilon-1<0$.

Thus, there is no value of $a$ for which both posterior means $x_\varepsilon$ and $y_\varepsilon$ are induced by the same interval-partitional signaling scheme. Moreover, the largest possible interval inducing a posterior mean that provides positive utility to the sender has length $1/2+2\varepsilon$. See Figure~\ref{fig:two-signal-upper-bound} for a visualization of the two optimal signaling schemes and their induced posterior means.
Therefore,
\[
\OPTpart^{u_\varepsilon}(2)=\frac12+2\varepsilon.
\]
Letting $\varepsilon\downarrow0$ proves that $c_2^{\mathrm{unif}}=1/2.$
\end{proof}

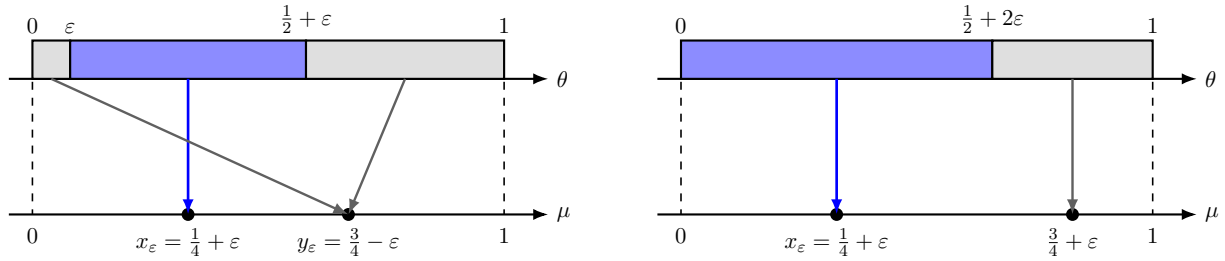
\begin{figure}[h]
    \centering
    \def\eps{0.08}
\def\W{8}

\begin{tikzpicture}[
    scale=0.78,
    transform shape,
    >=Latex,
    axis/.style={black, line width=0.8pt, -{Latex[length=2mm]}},
    boundary/.style={black, line width=0.8pt},
    guide/.style={black, dashed, line width=0.6pt},
    blueblock/.style={fill=blue!45},
    grayblock/.style={fill=gray!25},
    bluearrow/.style={blue, line width=1pt, -{Latex[length=2mm]}},
    grayarrow/.style={
        gray!75!black,
        line width=0.9pt,
        -{Latex[length=2mm]}
    },
    point/.style={circle, fill=black, inner sep=2.2pt}
]


\begin{scope}[xshift=0cm]

\draw[axis] (-0.4,1.4) -- (\W+0.75,1.4)
node[right] {$\theta$};

\draw[boundary,grayblock]
(0,1.4) rectangle ({\W*\eps},2.05);

\draw[boundary,blueblock]
({\W*\eps},1.4)
rectangle ({\W*(0.5+\eps)},2.05);

\draw[boundary,grayblock]
({\W*(0.5+\eps)},1.4)
rectangle (\W,2.05);

\node[above] at (0,2.05) {$0$};
\node[above] at ({\W*\eps},2.05) {$\varepsilon$};
\node[above] at ({\W*(0.5+\eps)},2.05)
{$\frac12+\varepsilon$};
\node[above] at (\W,2.05) {$1$};

\draw[axis] (-0.4,-0.9) -- (\W+0.75,-0.9)
node[right] {$\mu$};

\draw[guide] (0,1.35) -- (0,-0.9);
\draw[guide] (\W,1.35) -- (\W,-0.9);

\coordinate (xepsA) at ({\W*(0.25+\eps)},-0.9);
\coordinate (yepsA) at ({\W*(0.75-\eps)},-0.9);

\node[point] at (xepsA) {};
\node[point] at (yepsA) {};

\node[below=3pt] at (xepsA)
{$x_\varepsilon=\frac14+\varepsilon$};

\node[below=3pt] at (yepsA)
{$y_\varepsilon=\frac34-\varepsilon$};

\node[below=3pt] at (0,-0.9) {$0$};
\node[below=3pt] at (\W,-0.9) {$1$};

\draw[bluearrow]
({\W*(0.25+\eps)},1.4)
-- (xepsA);

\draw[grayarrow]
({\W*\eps/2},1.4)
-- (yepsA);

\draw[grayarrow]
({\W*(0.75+\eps/2)},1.4)
-- (yepsA);

\end{scope}


\begin{scope}[xshift=11cm]

\draw[axis] (-0.4,1.4) -- (\W+0.75,1.4)
node[right] {$\theta$};

\draw[boundary,blueblock]
(0,1.4)
rectangle ({\W*(0.5+2*\eps)},2.05);

\draw[boundary,grayblock]
({\W*(0.5+2*\eps)},1.4)
rectangle (\W,2.05);

\node[above] at (0,2.05) {$0$};
\node[above] at ({\W*(0.5+2*\eps)},2.05)
{$\frac12+2\varepsilon$};
\node[above] at (\W,2.05) {$1$};

\draw[axis] (-0.4,-0.9) -- (\W+0.75,-0.9)
node[right] {$\mu$};

\draw[guide] (0,1.35) -- (0,-0.9);
\draw[guide] (\W,1.35) -- (\W,-0.9);

\coordinate (xepsB) at ({\W*(0.25+\eps)},-0.9);
\coordinate (graymeanB) at ({\W*(0.75+\eps)},-0.9);

\node[point] at (xepsB) {};
\node[point] at (graymeanB) {};

\node[below=3pt] at (xepsB)
{$x_\varepsilon=\frac14+\varepsilon$};


\node[below=3pt] at (graymeanB)
{$\frac34+\varepsilon$};

\node[below=3pt] at (0,-0.9) {$0$};
\node[below=3pt] at (\W,-0.9) {$1$};

\draw[bluearrow]
({\W*(0.25+\eps)},1.4)
-- (xepsB);

\draw[grayarrow]
({\W*(0.75+\eps)},1.4)
-- (graymeanB);

\end{scope}

\end{tikzpicture}
    \caption{On the left, an optimal unrestricted scheme assigns the interval $[\varepsilon,1/2+\varepsilon]$ to one signal and its complement to the other, thereby inducing posterior means $x_\varepsilon=1/4+\varepsilon$ and $y_\varepsilon=3/4-\varepsilon$, respectively. On the right, an optimal interval-partitional scheme uses the threshold $1/2+2\varepsilon$, inducing posterior means $x_\varepsilon=1/4+\varepsilon$ and $3/4+\varepsilon$. }
    \label{fig:two-signal-upper-bound}
\end{figure}

\section{The case of three or more signals}\label{sec:kge3}
We characterize the uniform price of explainability for $K \ge 3$ in the following theorem.

\begin{theorem}[{Price of explainability for $K \ge 3$}]\label{thm:kge3}
For every $K\ge3$ and every bounded utility $u:[0,1]\to[0,1]$,
\[
    \OPTpart^u(K)\ge\frac23\OPT^u(K).
\]
The factor $2/3$ is tight for every $K\ge3$.  Consequently,
$c_K^{\rm unif}=2/3$ for every $K\ge3.$
\end{theorem}
In the following, we provide an overview of the proof. Fix an arbitrary signaling scheme with at most $K$ signals. Let
$S=\{x_1,\ldots,x_m\},$ with $m\leq K$ and $0 < x_1 < \cdots < x_m < 1,$ be the set of distinct
posterior means induced by the signaling scheme, and let
$p\in\mathbb{R}_+^m$ be the vector defined so that $p_i$ is the
probability that the posterior mean $x_i$ is realized.  Thus, the sender's expected utility achieved by this signaling scheme is equal to $\sum_{i=1}^{m} p_i u(x_i).$
Under the
uniform-prior assumption, $p$ must satisfy a collection of necessary
conditions in order to be induced by a signaling scheme. We denote by $F_m(S)$ the
convex set of probability vectors satisfying these conditions (see Definition~\ref{def:outer} for the formal definition). Since
every signaling scheme with posterior means in $S$ induces a vector
in $F_m(S)$, this set provides an outer relaxation. 

Similarly, every interval-partitional signaling scheme with at most
$K$ intervals induces a vector $w\in\mathbb{R}_+^m$, whose $i$-th
coordinate is the total length of the partition intervals centered at
$x_i$. This coordinate is zero if no such interval is present. We
denote by $W_K(S)$ the set of all vectors induced in this way. Since $W_K(S)$ is generally nonconvex, we consider its downward
convex hull $\downconv W_K(S)$, namely, the set of all nonnegative
vectors that are coordinatewise dominated by some vector in
$\conv W_K(S)$. Our main goal is to prove that, for
every $p\in F_m(S)$,
\begin{equation}\label{eq:defdownconv}
    \frac{2}{3}p\in\operatorname{downconv} W_K(S).
\end{equation}
Indeed, if Equation~\eqref{eq:defdownconv} holds, then Lemma~\ref{lem:payoff} in
Section~\ref{sec:posterior} implies that there exists a vector in $W_K(S)$, induced
by an interval-partitional signaling scheme achieving utility at least
two thirds of $\sum_{i=1}^{m} p_i u(x_i).$ Therefore, the proof of
Theorem~\ref{thm:kge3} reduces to establishing Equation~\eqref{eq:defdownconv}.

To prove Equation~\eqref{eq:defdownconv}, we use centered packings as an
intermediate object. A centered packing is a vector
$\lambda\in\mathbb{R}_+^m$, where the coordinate $\lambda_i$
represents the length of an interval centered at $x_i$, and
$\lambda_i=0$ when no such interval is present. The intervals
corresponding to the positive coordinates of $\lambda$ are contained
in $[0,1]$ and have pairwise disjoint interiors. Unlike a vector
$w\in W_K(S)$, which is induced by a complete partition of $[0,1]$,
a centered packing may leave parts of the state space uncovered. Whenever
a centered packing is random, we denote it by $\Lambda$.

As a first step, Lemma~\ref{lem:sparse-hull} in
Section~\ref{sec:packings} shows that every finitely supported random
centered packing $\Lambda$ such that every realization uses intervals
centered at no more than $K-1$ elements of $S$ satisfies
\begin{equation}\label{eq:expected-packing-in-downconv}
    \E[\Lambda]\in\downconv W_K(S).
\end{equation}
Therefore, in the remainder of the proof, we focus on constructing a
finitely supported random centered packing $\Lambda$ such that every
realization uses intervals centered at no more than $m-1$ elements
of $S$. Since $m\leq K$, the hypothesis of
Lemma~\ref{lem:sparse-hull} is then satisfied, and
Equation~\eqref{eq:expected-packing-in-downconv} follows.

For $m\geq 3$, the main technical step shows that, for every
$p\in F_m(S)$, there exists a finitely supported random centered
packing $\Lambda$ such that
\begin{equation}\label{eq:coordinatewise-packing-domination}
    \mathbb{E}[\Lambda_i]\geq 2p_i/3
\end{equation}
for every $i$, and every realization of $\Lambda$ uses intervals centered at no more than $m-1$ elements of $S$. This result is established in
Theorem~\ref{thm:sparse-domination} in Section~\ref{sec:induction} through a recursive
argument on the number of elements of $S$. Finally, the definition of $\operatorname{downconv} W_K(S)$, together with
Equation~\eqref{eq:expected-packing-in-downconv} and
Equation~\eqref{eq:coordinatewise-packing-domination}, implies
Equation~\eqref{eq:defdownconv}.

\subsection{Posterior-mean distributions}\label{sec:posterior}
Let $ S=\{x_1,\ldots,x_m\}\subset (0,1) $. Any partition of
$[0,1]$ into at most $K$ intervals induces a vector
$w\in\mathbb{R}_+^m$, where $w_i$ is the total length of the partition
intervals whose midpoint is $x_i$, and $w_i=0$ if no such interval is
present. We denote by $W_K(S)$ the set of all vectors induced by such
partitions. Thus, every vector $w \in W_K(S)$ is induced by at least one interval-partitional signaling scheme using at most $K$ intervals.
Furthermore, for any $A\subseteq\mathbb{R}_+^m$, we define
\[
\operatorname{downconv} A
:=
\left\{
z\in\mathbb{R}_+^m :
z\leq y \text{ coordinatewise for some } y\in\operatorname{conv} A
\right\}.
\]

In what follows, the set $\operatorname{downconv} W_K(S)$, rather than
$W_K(S)$ itself, will play a central role. First,
$\operatorname{downconv} W_K(S)$ is convex, whereas $W_K(S)$ need not be,
as illustrated in Figure~\ref{fig:WkSet}. Second, every vector
$y \in \operatorname{conv} W_K(S)$ can be interpreted as the average vector
generated by randomizing over interval partitions. Therefore, if $y\in \operatorname{conv} W_K(S)$ dominates a vector
$z\in \operatorname{downconv} W_K(S)$ coordinatewise, then, for every nonnegative utility $u$, the sender's
expected payoff under this randomization is at least the payoff associated with
$z$. We formalize this observation in the following lemma.

\begin{lemma}
\label{lem:payoff}
If $\alpha p \in \operatorname{downconv} W_K(S)$ for some $\alpha \geq 0$, then, for every nonnegative utility $u\colon [0,1]\to[0,1]$, there exists an interval-partitional signaling scheme using at most $K$ intervals whose expected sender utility is at least $\alpha \sum_{i=1}^m p_i u(x_i).$
\end{lemma}

\begin{proof}
There is a finite convex combination of vectors in $ W_K(S)$ whose
average $y$ satisfies $y_i\ge\alpha p_i$ for every $i$.  The average contribution from partition intervals whose midpoints belong to $S$ is at least
\[
    \sum_i y_i u(x_i)
    \ge
    \alpha\sum_i p_i u(x_i).
\]
Intervals whose midpoints are outside $S$ contribute additional
nonnegative payoff.  Hence one deterministic partition in the convex
combination attains at least the displayed value.
\end{proof}

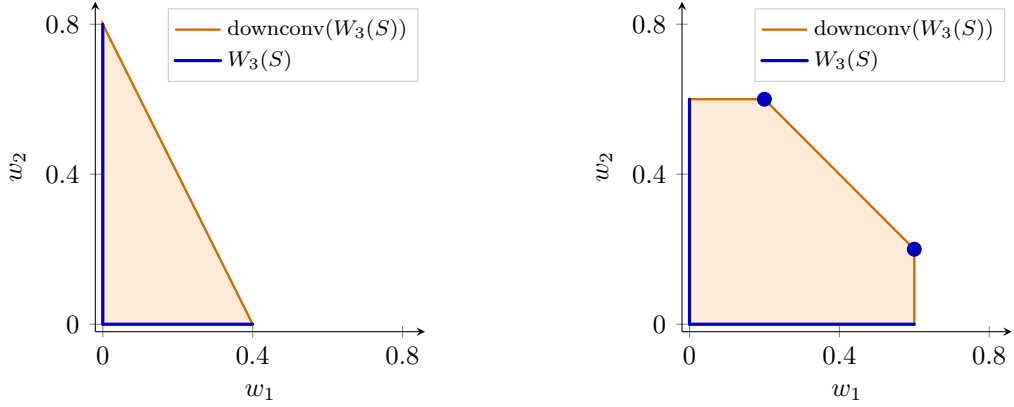
\begin{figure}[t]
    \centering
    \begin{tikzpicture}
\begin{groupplot}[
    group style={group size=2 by 1, horizontal sep=3.4cm},
    width=0.5\textwidth,
    height=0.36\textwidth,
    axis equal image,
    xmin=-0.02, xmax=0.86,
    ymin=-0.02, ymax=0.86,
    axis lines=left,
    xlabel={$w_1$},
    ylabel={$w_2$},
    xlabel style={font=\small},
    ylabel style={font=\small},
    tick label style={font=\small},
    xtick={0,0.4,0.8},
    ytick={0,0.4,0.8},
    clip=false,
]

\nextgroupplot[
    title style={font=\small},
    legend style={
        at={(0.98,0.98)},
        anchor=north east,
        font=\scriptsize,
        draw=gray!40,
        fill=white,
        fill opacity=0.85,
        text opacity=1,
        inner sep=2pt
    },
    legend cell align=left,
]

\addplot[
    draw=orange!85!black,
    fill=orange!70,
    fill opacity=0.22,
    line width=0.9pt
] coordinates {
    (0,0) (0.4,0) (0,0.8) (0,0)
};

\addplot[
    very thick,
    blue!75!black,
    line cap=round
] coordinates {(0,0) (0.4,0)};

\addplot[
    very thick,
    blue!75!black,
    line cap=round
] coordinates {(0,0) (0,0.8)};

\addlegendimage{very thick, blue!75!black}
\addlegendentry{$\operatorname{downconv}(W_3(S))$}
\addlegendimage{area legend, draw=orange!85!black, fill=orange!70, fill opacity=0.22}
\addlegendentry{$W_3(S)$}

\nextgroupplot[
    title style={font=\small},
    legend style={
        at={(0.98,0.98)},
        anchor=north east,
        font=\scriptsize,
        draw=gray!40,
        fill=white,
        fill opacity=0.85,
        text opacity=1,
        inner sep=2pt
    },
    legend cell align=left,
]

\addplot[
    draw=orange!85!black,
    fill=orange!70,
    fill opacity=0.22,
    line width=0.9pt
] coordinates {
    (0,0) (0.6,0) (0.6,0.2) (0.2,0.6) (0,0.6) (0,0)
};

\addplot[
    very thick,
    blue!75!black,
    line cap=round
] coordinates {(0,0) (0.6,0)};

\addplot[
    very thick,
    blue!75!black,
    line cap=round
] coordinates {(0,0) (0,0.6)};

\addplot[
    only marks,
    mark=*,
    mark size=2.6pt,
    blue!75!black
] coordinates {(0.6,0.2) (0.2,0.6)};

\addlegendimage{very thick, blue!75!black}
\addlegendentry{$\operatorname{downconv}(W_3(S))$}
\addlegendimage{area legend, draw=orange!85!black, fill=orange!70, fill opacity=0.22}
\addlegendentry{$W_3(S)$}

\end{groupplot}
\end{tikzpicture}
    \caption{The sets $W_3(S)$ (blue) and their downward convex hulls (orange) for $S=\{1/5,2/5\}$ (left) and $S=\{3/10,7/10\}$ (right).}
 \label{fig:WkSet}
\end{figure}

Given a finite-signal signaling scheme, discard zero-probability signals
and merge signals inducing the same posterior mean. Let
$
0<x_1<\cdots<x_m<1
$
be the resulting distinct posterior means. The strict inequalities follow
from the atomlessness of the prior $\operatorname{Unif}[0,1]$. Then the signaling scheme induces the posterior-mean distribution
\[
    Q=\sum_{i=1}^m p_i\delta_{x_i},
\]
where $\delta_x$ denotes the Dirac probability measure at $x$, and $p_i >0$ is the probability that the induced posterior mean equals $x_i$.

We next identify the feasibility information about $Q$ that the proof
will use.  Bayes plausibility fixes the total mass and the overall mean, but
those two equalities ignore the order of the state space.  Under the uniform prior, part of the order information is captured by the following necessary prefix inequalities.  Define the prefix masses,
first moments, and slacks by
\[
    A_j:=\sum_{i=1}^j p_i,
    \qquad
    M_j:=\sum_{i=1}^j p_i x_i,
    \qquad
    g_j:=M_j-\frac{A_j^2}{2}.
\]

Here $p_i$ is the probability of posterior mean $x_i$.  The value of the
scheme is $\sum_i p_i u(x_i)$. We also set $A_0=M_0=g_0:=0$.  The next lemma shows that all implementable
posterior-mean distributions have nonnegative proper-prefix slacks.

\begin{lemma}[Prefix inequalities]\label{lem:prefix}
Every posterior-mean distribution induced by a signaling scheme under
$\Unif[0,1]$ satisfies
\[
    \sum_{i=1}^m p_i=1,
    \qquad
    \sum_{i=1}^m p_i x_i=\frac12,
\]
and, for every $j<m$,
$
    M_j\ge{A_j^2}/{2}.
$
\end{lemma}

\begin{proof}
The first two identities are Bayes plausibility~\citep{KamenicaGentzkow2011}.  Fix $j<m$, and let
$q_j(\theta)$ be the probability that the signal sent at state $\theta$
induces one of the posterior means $x_1,\ldots,x_j$.  Then
\[
    0\le q_j\le1,
    \qquad
    \int_0^1q_j(\theta)\diff \theta=A_j,
    \qquad
    \int_0^1\theta q_j(\theta)\diff \theta=M_j.
\]
The rearrangement argument in Lemma~\ref{lem:max-mass} gives
\[
    M_j\ge\int_0^{A_j}\theta\diff \theta=\frac{A_j^2}{2},
\]
concluding the proof.
\end{proof}

An exact characterization of implementability is unnecessary for our
argument.  We retain only the necessary conditions in the preceding lemma and
prove the stronger domination theorem on the larger set that they define.
This outer relaxation has the convex structure needed by the induction.

\begin{definition}[Outer feasibility]\label{def:outer}
Fix an ordered target set $ S = \{x_1,\ldots,x_m\},$ with $ 0 < x_1 < \cdots < x_m < 1. $ A vector $p \in \mathbb{R}_+^m$ is \emph{outer feasible} if
\[
    \sum_{i=1}^m p_i=1,
    \qquad
    \sum_{i=1}^m p_ix_i=\frac12,
    \qquad
    g_j\ge0
    \quad(j<m).
\]
We denote by $F_m(S)$ the set of all outer-feasible probability vectors
indexed by $S$, allowing zero coordinates. For an interval $I=[\ell,\ell+h]$ with $h>0$, a finite measure
\[
    Q_I=\sum_{i=1}^m p_i\delta_{y_i},
    \qquad
    \ell<y_1<\cdots<y_m<\ell+h,
    \qquad
    p_i>0,
    \qquad
    \sum_{i=1}^m p_i=h,
\]
is outer feasible on $I$ if the probability measure
\[
    \widehat Q_I
    :=
    \sum_{i=1}^m\frac{p_i}{h}
    \delta_{(y_i-\ell)/h}
\]
is outer feasible on $[0,1]$.
\end{definition}

Every posterior-mean distribution induced by a signaling scheme is outer
feasible by Lemma~\ref{lem:prefix}.  From this point onward, the geometric
argument is carried out for all outer-feasible vectors, so no further
properties of signaling schemes are used until the final payoff reduction.

\subsection{Centered packings}\label{sec:packings}

In this section, we show how centered packings can be used in place of
$K$-interval partitions. More precisely, we show that every sufficiently
sparse centered packing, namely, every centered packing whose support has
cardinality at most $K-1$, belongs to
$\downconv W_K(S)$. Thus, although centered packings are more
flexible objects and may leave parts of $[0,1]$ uncovered, they can still
be coordinatewise dominated by convex combinations of vectors induced by
interval partitions using at most $K$ intervals. This allows us to carry
out the geometric construction using centered packings and return to interval
partitions only at the end.

Let $I=[\ell,\ell+h]$, where $h>0$ and $0\leq \ell<\ell+h\leq 1$, and let
$y_1<\ldots<y_r$ be points in the interior of $I$. A \emph{centered packing} in $I$ for the set $\{y_1,\ldots,y_r\}$ is a
vector $\lambda\in\mathbb{R}_+^r$ such that the intervals
\[
    \left[
        y_i-\frac{\lambda_i}{2},
        y_i+\frac{\lambda_i}{2}
    \right],
    \qquad i=1,\ldots,r,
\]
are contained in $I$ and have pairwise disjoint interiors. The support of a
centered packing $\lambda$ is
$\operatorname{supp}(\lambda) :=\{i\in\{1,\ldots,r\}:\lambda_i>0\}.$
The zero vector is also a valid centered packing. 

We denote by $\Lambda$ a random centered packing in $I$. We say that $\Lambda$ has a \emph{finitely supported distribution} if there exists a finite collection $\mathcal{C}$ of centered packings in $I$ such that $\Pr(\Lambda\in\mathcal{C})=1. $ Furthermore, let $T_I\colon I\to[0,1]$ be the affine map defined by
\[
    T_I(t):=\frac{t-\ell}{h}.
\]
Then a centered packing $\widehat{\lambda}\in\mathbb{R}_+^r$ in
$[0,1]$ for the target points
$T_I(y_1),\ldots,T_I(y_r)$ is mapped to the centered packing
$\lambda:=h\widehat{\lambda}$
in $I$ for the target points $y_1,\ldots,y_r$.

The next lemma establishes a connection between centered packings in $[0,1]$ and interval partitions. It shows that every centered packing supported on at most $K-1$ elements of $S$ belongs to $\operatorname{downconv} W_K(S)$. Consequently, the expectation of any random centered packing with a
finitely supported distribution that satisfies the same support bound
almost surely also belongs to $\downconv W_K(S)$.

\begin{lemma}\label{lem:sparse-hull}
Let $S=\{x_1,\ldots,x_m\}\subset (0,1)$. Every centered packing $\lambda$ in $[0,1]$ satisfying
\[
    |\supp(\lambda)|\leq K-1
\]
belongs to $\downconv W_K(S)$. Consequently, if $\Lambda$ is a random centered packing in $[0,1]$ with a finitely supported distribution and $ \lvert\operatorname{supp}(\Lambda)\rvert\leq K-1 $ \text{almost surely}, then $ \mathbb{E}[\Lambda] \in \operatorname{downconv} {W}_K(S). $
\end{lemma}
\begin{proof}
Let $\lambda\in\R_+^m$ be a centered packing in $[0,1]$ such that
$
    |\supp(\lambda)|\leq K-1.
$
If $\lambda$ is the zero vector, the claim is immediate. Otherwise, write
\[
    J:=\supp(\lambda)=\{j_1<\cdots<j_r\},
\]
so that $r\leq K-1$, and set
$y_q:=x_{j_q} $, for every $ q=1,\ldots,r.$
The set of all centered packings in $[0,1]$ for the target points
$y_1,\ldots,y_r$ is the bounded polytope
\[
\begin{split}
    P_J:=\bigl\{a\in\R_+^r:\;&
    a_1\leq 2y_1,\quad
    a_q+a_{q+1}\leq 2(y_{q+1}-y_q)
    \ \text{for }\, q\le r-1,\quad
    a_r\leq 2(1-y_r)
    \bigr\}.
\end{split}
\]
The polytope $P_J$ is full-dimensional, since sufficiently small positive
values of $a_1,\ldots,a_r$ satisfy all its defining inequalities strictly.
Since a bounded polytope is the convex hull of its extreme points and
$\downconv W_K(S)$ is convex, it is enough to consider an extreme
point $a\in P_J$.

Suppose that exactly $k$ of the $r$ coordinates of $a$ are positive.
Then $r-k$ nonnegativity constraints are active. Since $P_J$ is
full-dimensional, an extreme point must satisfy $r$ linearly independent
active constraints. Hence at least $k$ of the remaining $r+1$ constraints
defining $P_J$ are active.

We call a {gap} any portion of $[0,1]$ left uncovered between two
consecutive centered intervals, or between an endpoint of $[0,1]$ and the
nearest centered interval. The lengths of the $r+1$ gaps determined by
$a$ are
\[
\begin{aligned}
    \gamma_0
        &:= y_1-\frac{a_1}{2},\\
    \gamma_q
        &:= y_{q+1}-y_q-\frac{a_q+a_{q+1}}{2},
        && q=1,\ldots,r-1,\\
    \gamma_r
        &:= 1-y_r-\frac{a_r}{2}.
\end{aligned}
\]
An active constraint is equivalent to the corresponding gap having length
zero. Therefore, at least $k$ of the $r+1$ gaps have length zero, and at
most $r+1-k$ gaps have positive length.

The $k$ positive centered intervals, together with the gaps of positive
length, form a partition of $[0,1]$ into at most
\[
    k+(r+1-k)=r+1\leq K
\]
nonempty intervals. This partition induces a vector
$w\in W_K(S)$. For every $q=1,\ldots,r$, the centered interval
of length $a_q$ has midpoint $y_q=x_{j_q}$, while the added gaps may
contribute additional nonnegative lengths at other targets in $S$.
Therefore, after restoring zero coordinates outside $J$, $w$
coordinatewise dominates $a$. Hence every extreme point of $P_J$
belongs to $\downconv W_K(S)$, and so does every point of $P_J$
by convexity.

Since $(\lambda_{j_1},\ldots,\lambda_{j_r})\in P_J,$
the first claim follows. For the second claim, every realization of
$\Lambda$ belongs to $\downconv W_K(S)$ by the first claim.
Since $\Lambda$ has a finitely supported distribution and
$\downconv W_K(S)$ is convex,
$
    \E[\Lambda]\in\downconv W_K(S).
$
\end{proof}

Consequently, it remains to construct a random centered packing $\Lambda$ in $[0,1]$, with a finitely supported distribution, such that $ \mathbb{E}[\Lambda_i]\geq {2p_i}/{3}, $ for every $i=1,\ldots,m, $ and $ \lvert\operatorname{supp}(\Lambda)\rvert\leq m-1$ almost surely.

\subsection{The one- and two-point building blocks}\label{sec:blocks}

A tight prefix will split an outer-feasible measure into adjacent blocks.
The induction therefore needs explicit certificates for blocks of sizes one
and two.  Besides the two-thirds expectation, we track the event on which a
block uses all of its targets.  These exceptional probabilities are chosen so
that the certificates of two adjacent blocks can later be coupled without
both being exceptional simultaneously.

For a one-point block, outer feasibility fixes the target at the midpoint of
the block.  We deliberately use the whole block only with probability
$2/3$; the remaining probability creates the slack needed when this block is
paired with a two-point block.

\begin{lemma}[One-point block]\label{lem:one-point}
Let $Q_I=h\delta_y$ be outer feasible on
$I=[\ell,\ell+h]$.  There is a finitely supported random centered packing $\Lambda$ in
$I$ such that
\[
    \E[\Lambda]=\frac23h,
\]
and $ \Pr\bigl(|\supp(\Lambda)|=1\bigr)=2/3.$
On the complementary event the packing is empty.
\end{lemma}

\begin{proof}
Outer feasibility forces $y=\ell+h/2$.  Take the whole interval $I$ with
probability $2/3$, and the empty packing with probability $1/3$.
\end{proof}

The two-point block is the only algebraic building block.  We need to preserve
exactly two thirds of both masses while ensuring that both targets are used
together with probability only $1/3$.  This $1/3$ matches the probability
with which the preceding one-point certificate is empty and is what makes a
$1+2$ split possible.

\begin{lemma}\label{lem:two-point}
Let $Q_I=p_1\delta_{y_1}+p_2\delta_{y_2}$ be a positive two-point outer feasible measure on an interval
$I=[\ell,\ell+h]$.  There is a finitely supported random centered packing
$\Lambda=(\Lambda_1,\Lambda_2)$ in $I$ such that
\[
    \E[\Lambda_i]=\frac23p_i,\quad \forall i=1,2,
\]
and $\Pr\bigl(|\supp(\Lambda)|=2\bigr)=1/3.$
\end{lemma}

\begin{proof}
By affine rescaling, it is enough to consider $I=[0,1]$.  Let
\[
    Q=\pi_1\delta_a+\pi_2\delta_b,
    \quad \textnormal{where} \quad
    a=\frac12-u<\frac12<b=\frac12+v,
\]
for some $ u,v>0 $ and set $d:=u+v=b-a.$ Bayes plausibility requires $\pi_1 a+\pi_2 b=1/2$. Substituting $a=1/2-u$ and $b=1/2+v$ and observing that $\pi_1+\pi_2=1 $, we obtain
\[
    -u\pi_1+v\pi_2=0.
\]
Hence, recalling that $d=u+v$,
we have $\pi_1={v}/{d}$, and $ \pi_2={u}/{d }.$
The first prefix inequality is $\pi_1\le2a$, or equivalently
\[
    \frac vd\le1-2u
    \quad\Longleftrightarrow\quad
    u(1-2d)\ge0.
\]
Since $u>0$, this implies $d\le1/2$.
Set $A:=2a=1-2u,$ and $B:=2(1-b)=1-2v.$ Then $A,B>0$, and $(A,0)$ and $(0,B)$ are centered packings with support size $1$.  Furthermore, the vector $q=(2v,2u)$ corresponds to the intervals
\[
    I_a=[a-v,a+v]
    \qquad\text{and}\qquad
    I_b=[b-u,b+u],
\]
centered at $a$ and $b$, respectively. These intervals touch, since $ a+v=b-u.$
Moreover, they are both contained in $[0,1]$. Indeed, since
$d=u+v\leq 1/2$,
\[
    2v\leq 1-2u=A
    \qquad\text{and}\qquad
    2u\leq 1-2v=B.
\]
Equivalently, $ a-v\geq 0$ and $b+u\leq 1.$
Therefore, $q$ is a centered packing with support size $2$.
Use $q$ with probability $1/3$, use $(A,0)$ with probability
\[
    s_1:=\frac{2(1-d)\pi_1}{3A},
\]
and use $(0,B)$ with probability
\[
    s_2:=\frac{2(1-d)\pi_2}{3B}.
\]
To verify that the assigned probabilities are well defined, it is sufficient to show that the following inequality is satisfied:
\[
1-(1-d)\left(\frac{\pi_1}{A}+\frac{\pi_2}{B} \right) \ge 0.
\]
After substituting the definitions of $\pi_1,\pi_2,A,$ and $B$, and bringing
this expression to the common denominator
$d(1-2u)(1-2v)$, its numerator is
\begin{align*}
    d(1-2u)(1-2v)
    -(1-d)\bigl[v(1-2v)+u(1-2u)\bigr]=(1-2d)(u-v)^2,
\end{align*}
where $d=u+v$.  Therefore
\[
    1-(1-d)\left(\frac{\pi_1}{A}+\frac{\pi_2}{B}\right)
    =
    \frac{(1-2d)(u-v)^2}
    {d(1-2u)(1-2v)}
    \ge0
\]
implies $ 1/3+s_1+s_2\le1.$ Assign the remaining probability to the empty packing.  Finally,
\[
    s_1(A,0)+s_2(0,B)
    =
    \frac{2(1-d)}{3}(\pi_1,\pi_2),
\]
and therefore
\[
    \frac13q+s_1(A,0)+s_2(0,B)
    =
    \frac23(\pi_1,\pi_2).
\]
Thus the expected lengths are exactly two thirds of the target masses, and the
only outcome using both targets has probability $1/3$.  Under the inverse
affine map, lengths are multiplied by $h$ and
$p_i=h\pi_i$, so the expected original-scale lengths are
$h(2\pi_i/3)=2p_i/3$.  This proves the result on~$I$.
\end{proof}

When two block certificates are concatenated independently, they may both use
all of their own targets, causing the combined packing to use every target of
the original instance.  The next elementary coupling lemma prevents this.  It
shows that whenever the two exceptional probabilities sum to at most one, the
marginal certificate laws can be preserved while making the exceptional events
disjoint.

\begin{lemma}[Disjoint-exception coupling]\label{lem:coupling}
Let $I_L$ and $I_R$ be adjacent intervals.  For $b\in\{L,R\}$, let
$\Lambda^b$ be a finitely supported random centered packing in $I_b$ on
$r_b$ targets lying in the interior of $I_b$, and let $E_b$ be an event
such that
\[
    |\supp(\Lambda^b)|\le r_b-1
    \quad\text{on }E_b^c.
\]
If
\[
    \Pr(E_L)+\Pr(E_R)\le1,
\]
then the two packings can be coupled, without changing either marginal law, so
that their concatenation is a centered packing and
\[
    |\supp(\Lambda^L)|+|\supp(\Lambda^R)|
    \le r_L+r_R-1
\]
almost surely.
\end{lemma}

\begin{proof}
Let $e_b=\Pr(E_b)$, and let $\mu_b^1$ and $\mu_b^0$ denote the
conditional laws of $\Lambda^b$ on $E_b$ and $E_b^c$, respectively.
If a conditioning event has probability zero, choose the corresponding
conditional law arbitrarily.
Its coefficient in the mixture below is then zero.
Construct the joint law as the mixture
\[
    e_L(\mu_L^1\otimes\mu_R^0)
    +e_R(\mu_L^0\otimes\mu_R^1)
    +(1-e_L-e_R)(\mu_L^0\otimes\mu_R^0),
\]
with zero-weight terms omitted.  The left and right marginals are, respectively,
\[
    e_L\mu_L^1+(1-e_L)\mu_L^0
    \qquad\text{and}\qquad
    e_R\mu_R^1+(1-e_R)\mu_R^0,
\]
which are the original laws.  Declare the left block exceptional only in the
first mixture component and the right block exceptional only in the second;
these exceptional events are disjoint.  In every component at least one block
is drawn from $\mu_b^0$, where it uses at most $r_b-1$ targets, while the
other block uses at most all of its targets.  This gives the support bound.
Finally, the two target sets are strictly separated by the common block
boundary.  Since each packing is contained in its own block, their intervals
have disjoint interiors after concatenation, so the result is a centered
packing.
\end{proof}

The building blocks and the coupling explain how to combine two smaller
instances.  It remains to justify that the boundary structure of the
outer-feasible set always produces exactly such smaller adjacent instances.

\subsection{A unified sparse-domination induction}\label{sec:induction}

The recursive step is driven by tight prefix inequalities.  When
$g_j=0$, the first $j$ posterior means and the remaining posterior means
lie on opposite sides of the split point $A_j$, and after affine rescaling
each side satisfies the same outer-feasibility conditions.  The following
lemma records this self-similar decomposition.

\begin{lemma}[Decomposition at a tight prefix]\label{lem:prefix-split}
Let
\[
    Q=\sum_{i=1}^m p_i\delta_{x_i},
    \qquad
    0<x_1<\cdots<x_m<1,
    \qquad
    p_i>0,
\]
be outer feasible.  Suppose that, for some $1\le j<m$,
\[
    g_j=0.
\]
Set $L:=A_j$ and $R:=1-L$.  Then the first $j$ points form an outer
feasible measure on $[0,L]$, and the remaining $m-j$ points form an outer
feasible measure on $[L,1]$.
\end{lemma}

\begin{proof}
Outer feasibility gives
\[
    g_0=g_m=0,
    \qquad
    g_k\ge0\quad(k=1,\ldots,m-1).
\]
Since all masses are positive, $0<L<1$. The support separates at $L$.  Since $A_j=L$ and
$A_{j-1}=L-p_j$,
\[
    g_j-g_{j-1}
    =
    p_j\left(x_j-L+\frac{p_j}{2}\right).
\]
Because $g_j=0\le g_{j-1}$,
\[
    x_j\le L-\frac{p_j}{2}<L.
\]
Similarly,
\[
    g_{j+1}-g_j
    =
    p_{j+1}\left(x_{j+1}-L-\frac{p_{j+1}}{2}\right),
\]
where $g_m=0$ is used when $j=m-1$.  Since
$g_{j+1}\ge g_j=0$,
\[
    x_{j+1}\ge L+\frac{p_{j+1}}{2}>L.
\]
By monotonicity of the support, all left-block points lie in $(0,L)$, and all
right-block points lie in $(L,1)$.
Normalize the left block by dividing masses by $L$ and locations by $L$.
For every prefix ending at an original index $k\le j$, the normalized prefix
mass and first moment are $A_k/L$ and $M_k/L^2$.  Its slack is therefore
\[
    \frac{M_k}{L^2}
    -\frac12\left(\frac{A_k}{L}\right)^2
    =
    \frac{g_k}{L^2}\ge0.
\]
At $k=j$, this slack is zero and the normalized block has total mass one and
mean $1/2$.  Hence the left block is outer feasible on $[0,L]$.

Normalize the right block by translating locations by $L$, then dividing
masses and locations by $R$.  For every prefix ending at an original index
$k>j$, its normalized mass and first moment are
\[
    \frac{A_k-L}{R}
    \qquad\text{and}\qquad
    \frac{M_k-M_j-L(A_k-L)}{R^2}.
\]
Using $M_j=L^2/2$, the normalized slack is
\[
    \frac{M_k-M_j-L(A_k-L)}{R^2}
    -\frac12\left(\frac{A_k-L}{R}\right)^2
    =
    \frac{M_k-L^2/2-L(A_k-L)-(A_k-L)^2/2}{R^2}
    =
    \frac{g_k}{R^2}\ge0.
\]
At $k=m$, this slack is zero and the normalized block has total mass one and
mean $1/2$.  Hence the right block is outer feasible on $[L,1]$.
\end{proof}

We now have all ingredients for the central geometric statement.  It is
stronger than the desired payoff theorem in two ways: it preserves every
target mass coordinatewise, and it guarantees that each realized packing
omits at least one target.  The latter is precisely the hypothesis needed by
Lemma~\ref{lem:sparse-hull} to return to the original interval budget.

\begin{theorem}[Sparse domination]\label{thm:sparse-domination}
Fix $m\ge3$ and ordered targets
\[
    0<x_1<\cdots<x_m<1.
\]
For every outer feasible vector $p\in\R_+^m$, there is a finitely supported random
centered packing $\Lambda$ on these targets such that
\[
    \E[\Lambda_i]\ge\frac23p_i
    \qquad(i=1,\ldots,m),
\]
and
\[
    |\supp(\Lambda)|\le m-1
\]
almost surely.
\end{theorem}

\begin{proof}
The proof combines convex geometry with the tight-prefix recursion.  For a
fixed target support, we first reduce to extreme points of the outer-feasible
mass set.  An extreme point must either have a zero mass, which lowers the
support size, or a tight proper-prefix inequality, which invokes
Lemma~\ref{lem:prefix-split}.  In the latter case we certify the two blocks
recursively and use Lemma~\ref{lem:coupling} to glue them while omitting at
least one target.

We argue by strong induction on $m\ge3$.  Whenever the induction hypothesis is
used below, it is applied to the same statement on $r$ ordered interior
targets with $3\le r<m$; sizes one and two are handled by
Lemmas~\ref{lem:one-point} and~\ref{lem:two-point}.  In particular, the argument
also supplies the base case $m=3$.

\medskip\noindent\textbf{Reduction to extreme points.}
Fix the target points and let $\mathcal F_m$ be the set of outer feasible mass
vectors on this support.  If $\mathcal F_m$ is empty there is nothing to prove.
Otherwise it is compact and convex.  Indeed, its nonnegativity and total-mass
constraints place it in the compact probability simplex; the total-mass and
mean equalities are closed affine constraints; and each prefix constraint is a
closed superlevel set of the continuous concave function
\[
    p\longmapsto M_j(p)-\frac{A_j(p)^2}{2}.
\]
The set of mass vectors admitting the asserted random packing is also convex,
because mixing two finitely supported certificate laws preserves finite
support, the per-outcome support bound, and the coordinatewise expectation
inequalities.  Since $\mathcal F_m$ is compact and convex, Minkowski's theorem gives
\[
    \mathcal F_m
    =
    \conv\bigl(\operatorname{ext}(\mathcal F_m)\bigr),
\]
and Carath\'eodory's theorem implies that every point of $\mathcal F_m$ is a
convex combination of at most
$\dim(\operatorname{aff}\mathcal F_m)+1$ extreme points
\citep[Theorems~17.1 and~18.5]{Rockafellar1970}.
It is therefore enough to treat an extreme point $p\in\mathcal F_m$.

\medskip\noindent\textbf{Structure of an extreme point.}
At such an extreme point, either some mass is zero or some proper prefix
inequality is tight, with $1\le j<m$.  Indeed, suppose that every mass and
every proper prefix slack $g_j(p)$ were strictly positive.  Because $m\ge3$,
the common nullspace of the two linear functionals below contains a nonzero
vector $v$:
\[
    \sum_i v_i=0,
    \qquad
    \sum_i x_iv_i=0.
\]
The two equalities defining total mass and mean are therefore preserved along
$p+tv$.  Since there are only finitely many coordinates and proper prefixes,
coordinate positivity and continuity of every $g_j$ give an
$\varepsilon>0$ for which $p+tv$ has nonnegative coordinates and all proper
prefix slacks nonnegative whenever $|t|\le\varepsilon$.  (The final slack is
always zero under the preserved total-mass and mean equalities.)  Thus the two
distinct vectors $p+\varepsilon v$ and $p-\varepsilon v$ lie in
$\mathcal F_m$ and have midpoint $p$, contradicting extremality.

\medskip\noindent\textbf{Zero-mass case.}
Suppose first that some $p_i=0$.  Delete all zero coordinates, and let
$1\le r<m$ be the remaining support size.  Total mass and mean are unchanged.
Moreover, every prefix of the reduced ordered support has the same mass and
first moment as the original prefix ending at its last retained index, since
all skipped terms have zero mass.  Hence the reduced vector is outer feasible.
If $r=1$, use Lemma~\ref{lem:one-point}; if $r=2$, use
Lemma~\ref{lem:two-point}; and if $r\ge3$, use the induction hypothesis.
After embedding the resulting packing back into the original coordinate space
by assigning zero lengths to the deleted coordinates, its support size is at
most $r\le m-1$, and its expected lengths dominate $(2/3)p$.

\medskip\noindent\textbf{Tight-prefix case.}
Assume now that every mass is positive and that
\[
    M_j=\frac{A_j^2}{2}
\]
for some $j<m$.  By Lemma~\ref{lem:prefix-split}, the distribution splits
into adjacent outer feasible blocks of sizes
\[
    r_L:=j,
    \qquad
    r_R:=m-j.
\]
For a block of size one, use Lemma~\ref{lem:one-point} and declare its nonempty
outcome exceptional; its exceptional probability is $2/3$, and outside that
event it uses zero targets.  For a block of size two, use
Lemma~\ref{lem:two-point} and declare the two-target outcomes exceptional; their
total probability is exactly $1/3$, and outside that event the packing uses at
most one target.  For a block of size at least three, first normalize it to
$[0,1]$, use the induction hypothesis, undo the affine normalization, and take
the exceptional event to be empty; every outcome then uses at most one fewer
target than that block contains.  In every case the certificate is finitely
supported and has coordinate expectations at least two thirds of the block
masses.

Define
\[
    \eta_1:=\frac23,
    \qquad
    \eta_2:=\frac13,
    \qquad
    \eta_r:=0\quad(r\ge3).
\]
The exceptional probability of a block of size $r$ is at most $\eta_r$.
Since $r_L,r_R\ge1$ and $r_L+r_R=m\ge3$, the only potentially largest pair
is $\{r_L,r_R\}=\{1,2\}$, and hence
\[
    \eta_{r_L}+\eta_{r_R}\le1.
\]
Lemma~\ref{lem:coupling} therefore couples the two block certificates without
changing their marginals and makes their exceptional events disjoint.  After
undoing the affine normalizations and concatenating the two adjacent packings,
we obtain a centered packing on the original targets.  Its expected length in
every coordinate is at least two thirds of the corresponding mass because the
coupling preserves both marginal laws and affine rescaling multiplies masses and
lengths by the same block length.  In every coupled outcome at least one block
is nonexceptional, so the two blocks together use at most
\[
    r_L+r_R-1=m-1
\]
targets.

\medskip\noindent\textbf{Convexification.}
Thus every extreme point of $\mathcal F_m$ has the required certificate.  For
an arbitrary $p\in\mathcal F_m$, write $p$ as a finite convex combination of
extreme points and first sample the extreme point according to those weights,
then sample its certificate.  This produces a finitely supported law, preserves
the support bound in every outcome, and by linearity gives expected lengths at
least $(2/3)p$.  This completes the induction.
\end{proof}

\subsection{Completion of the proof}

The geometric work is now complete.  Starting from an unrestricted scheme,
Lemma~\ref{lem:prefix} places its posterior-mean mass vector in the
outer-feasible set.  Theorem~\ref{thm:sparse-domination} produces a sparse
packing with the required coordinatewise expectations,
Lemma~\ref{lem:sparse-hull} converts that packing into the original interval
budget, and Lemma~\ref{lem:payoff} returns to the designer's objective.  We now
carry out this chain and then recall the example giving tightness.

\begin{proof}[Proof of Theorem~\ref{thm:kge3}]
\medskip\noindent\textbf{Lower bound.}
Consider an arbitrary unrestricted scheme with at most $K$ signals, and
combine signals having the same posterior mean.  Let
\[
    Q=\sum_{i=1}^m p_i\delta_{x_i},
    \qquad
    m\le K,
\]
be the resulting posterior-mean distribution.  By Lemma~\ref{lem:prefix}, its
mass vector is outer feasible.

If $m=1$, take the full interval $[0,1]$ as a deterministic centered packing; its midpoint is $1/2$, and it matches the scheme.  If $m=2$, Lemma~\ref{lem:two-point} gives a random
centered packing that dominates $(2/3)p$ and is supported on at most two
targets.  If $m\ge3$, Theorem~\ref{thm:sparse-domination} gives such a random
packing supported on at most $m-1$ targets.  In every case there is a random
packing $\Lambda$ satisfying
\[
    \E[\Lambda_i]\ge\frac23p_i
\]
and
\[
    |\supp(\Lambda)|\le K-1
\]
almost surely: the three support bounds are respectively
$1\le K-1$, $2\le K-1$, and $m-1\le K-1$.  All three laws are finitely
supported.  Lemma~\ref{lem:sparse-hull} first gives
$\E[\Lambda]\in\downconv W_K(\{x_1,\ldots,x_m\})$.  Since this set is
coordinatewise downward closed and $\E[\Lambda]\ge(2/3)p$, it follows that
\[
    \frac23p
    \in
    \downconv W_K(\{x_1,\ldots,x_m\}).
\]
By Lemma~\ref{lem:payoff}, some deterministic $K$-interval partition has
value at least two thirds of the value of the unrestricted scheme.  Thus, for
the value $V$ of every individual admissible scheme,
$\OPTpart^u(K)\ge(2/3)V$.  Taking the supremum over schemes on the right---with
no assumption that the supremum is attained---proves the lower bound.

\medskip\noindent\textbf{Tightness.}
For tightness, let
\[
    u(t)=\mathbf 1\{t=1/3\}+\mathbf 1\{t=2/3\}.
\]
The unrestricted scheme that signals $[1/12,7/12]$ versus its complement
has first-signal mass $1/2$ and mean
$(1/12+7/12)/2=1/3$.  Its complement has first moment
$1/2-(1/2)(1/3)=1/3$, so its conditional mean is $2/3$.  Hence the two
signals each have probability $1/2$ and the scheme has value one.  Since
$0\le u\le1$, this also proves $\OPT^u(K)=1$ for every $K\ge3$.

In an interval partition, only intervals centered at $1/3$ or $2/3$
contribute.  Two distinct positive partition intervals cannot have the same
midpoint, so there is at most one useful interval at each target.  Let their
lengths be $\ell$ and $r$, with a missing interval interpreted as length
zero.  If both are present, disjoint interiors give
\[
    \frac13+\frac\ell2
    \le
    \frac23-\frac r2,
\]
while if either is missing the same inequality follows from containment of the
remaining centered interval in $[0,1]$.  Thus $\ell+r\le2/3$, and every
interval partition has value at most $2/3$.

The two-interval partition $[0,2/3]$, $[2/3,1]$ is admissible under the
``at most $K$'' convention for every $K\ge3$.  Its first interval has
midpoint $1/3$ and contributes $2/3$, while the second contributes zero.
(Under the displayed $K$-slot convention, repeat endpoints to add
zero-length intervals.)  Hence $\OPTpart^u(K)=2/3$ and the factor is tight for
every $K\ge3$.
\end{proof}

\section{Putting it all together}\label{sec:main-result}
\begin{theorem}\label{thm:main}
For one-dimensional linear information design under
$\theta\sim\Unif[0,1]$, when the utility ranges over all bounded functions
$u:[0,1]\to[0,1]$, the uniform price of explainability satisfies
\[
    c_K^{\rm unif}
    =
    \begin{cases}
        \dfrac12, & K=2,\\[2mm]
        \dfrac23, & K\ge3.
    \end{cases}
\]
\end{theorem}
\begin{proof}
Theorem~\ref{thm:k2} gives $c_2^{\rm unif}=1/2$, and
Theorem~\ref{thm:kge3} gives $c_K^{\rm unif}=2/3$ for every $K\ge3$.
\end{proof}

Thus, the $2/3$ guarantee requires no increase in the signal budget for
any $K\ge3$, while $K=2$ is the unique exception.

\section{Conclusions and future work}\label{sec:main-conclusion}

We have established the exact price of explainability under the uniform prior, showing that no additional signals are needed to attain the tight $2/3$ guarantee for every $K \ge 3$, while $K=2$ remains the unique exception. A natural direction for future work is to determine how the exact price of explainability depends on the prior. In particular, it would be interesting to identify structural properties of the prior that lead to guarantees strictly larger than the general $1/2$ bound. Another direction is to extend the analysis to higher-dimensional state spaces, where the appropriate notion of an explainable partition and the geometry of feasible posterior means are substantially more complex.

\clearpage
\bibliographystyle{plainnat}
\bibliography{references}
\end{document}